\documentclass[12pt,a4paper]{article}
\usepackage{amsthm,amssymb,amsmath}
\usepackage{xcolor,complexity,float}
\usepackage{graphicx}
\usepackage{thmtools}

\usepackage{hyperref,fullpage}
\newcommand{\var}{{\sf var}}

\newcommand{\lab}{{\sf label}}
\newcommand{\rank}{{\sf rank_{\varphi}}}

\newcommand{\lo}{{\cal L}}

\newcommand{\kub}{k\mbox{-}unbalanced}
\newcommand{\cintabp}{{\sf strict~circular\mbox{-}interval~ABP}}

\newtheorem{defn}{Definition}

\newtheorem{lemma}{Lemma}

\newtheorem{claim}{Claim}
\newtheorem{theorem}{Theorem}

\title{Lower bounds for multilinear   bounded  order ABPs}
\author{C. Ramya \hspace*{15 mm} B. V. Raghavendra Rao \\
IIT Madras, Chennai, India\\
\{ramya,bvrr\}@cse.iitm.ac.in}

\begin{document}

\maketitle


\begin{abstract}
Proving super-polynomial size lower bounds for syntactic multilinear Algebraic Branching  Programs(smABPs) computing an explicit polynomial is a challenging problem in Algebraic Complexity Theory. The order in which variables in $\{x_1,\ldots,x_n\}$ appear along a source to sink path in any smABP can be viewed as a permutation in $S_n$. In this article, we consider the following special classes of smABPs where the order of occurrence of variables along a source to sink path is restricted:


\begin{enumerate}
\item \textbf{Strict circular-interval ABPs:} For every subprogram the  index set of  variables occurring in it  is contained in some circular interval of $\{1,\ldots, n\}$. 
 \item  \textbf{$\lo$-ordered  ABPs:}  There is a set of $\lo$ permutations (orders) of variables  such that every source to sink path in the ABP reads variables in one of these $\lo$ orders.  
 \end{enumerate}
 
We prove exponential (i.e. $2^{\Omega(\sqrt{n}/\log n)}$) lower bound  for the  size of a  strict circular-interval ABP computing an explicit  multilinear $n$-variate polynomial in ${\sf VP}$. For the same polynomial, we show that any sum of $\lo$-ordered ABPs of small size will require exponential (i.e. $2^{n^{\Omega(1)}}$) many summands, when  $\lo \le 2^{n^{1/2-\epsilon}}$, $\epsilon>0$. 

At the heart of above lower bound arguments is a new decomposition theorem  for  smABPs: We show that  any polynomial that can be computed by an smABP of size $S$, can be written as a sum  of  $O(S)$ many multilinear polynomials where each summand is a product of two polynomials  in at most $2n/3$ variables, computable by smABPs.  As an immediate corollary to our decomposition theorem for smABPs, we obtain a low bottom fan-in version of the depth reduction by Tavenas~[MFCS, 2013] for the case of smABPs. In particular,  we show that a polynomial that has size $S$ smABPs can be  expressed as a sum of products of multilinear polynomials on $O(\sqrt{n})$ variables, where the total number of summands is bounded by $2^{O(\log n \log S\sqrt{n})}$. Additionally, we show that $\lo$-ordered ABPs can be transformed into $\lo$-pass smABPs with a polynomial blowup in size. 
\end{abstract}
\newpage 
\section{Introduction}
Algebraic Complexity Theory is concerned with classification of polynomials based on the  number of algebraic operations required to compute a polynomial from variables and constants. Arithmetic circuits, one of the most popular models for algebraic computation was introduced by Valiant~\cite{Val79}. Since their inception, arithmetic circuits have served as the primary model of computation  for polynomials. 

One of the primary tasks in Algebraic Complexity theory is proving lower bounds on the size of arithmetic circuits computing an explicit polynomial. Valiant~\cite{Val79} conjectured that the polynomial defined by the permanent of an $n\times n$ symbolic matrix is not computable by polynomial size arithmetic circuits, known  as Valiant's hypothesis and is one of the central questions  in  algebraic complexity theory.

The best known size lower bound for general classes of arithmetic circuits is only super-linear~\cite{BS83} in the number of variables. Despite several approaches,  the problem of  proving lower bounds for general classes of arithmetic circuits has remained elusive. Naturally, there have been efforts to prove lower bounds for special classes of arithmetic circuits which led to the development of several lower bound techniques. Structural restrictions such as depth and  fan-out, semantic restrictions such as multilinearity and homogeneity have received widespread attention in the literature.

 Agrawal and Vinay~\cite{AV08} showed that proving exponential lower bounds for depth four circuits is sufficient to prove Valiant's hypothesis. This initiated several attempts at proving lower bound for constant depth circuits. (See \cite{Sap} for a detailed survey of these results.)

Among other restrictions, multilinear circuits where every gate computes a multilinear polynomial have received wide attention.   Multilinear circuits are natural models for computing multilinear polynomials. In many situations, it is useful to consider a  natural  syntactic sub-class   of multilinear circuits.  A circuit is {\em syntactic multilinear} if the children of every product gate depend on disjoint sets of variables.  Raz~\cite{Raz09} obtained super-polynomial lower bounds for syntactic multilinear formulas computing the determinant or permanent polynomial which was further improved for constant depth multilinear circuits~\cite{RY09,CLS18,CELS18}.  However, the best known lower bound for syntactic multilinear circuits is only almost  quadratic~\cite{AKV18}.
 
Algebraic branching programs(ABPs) are special classes of arithmetic circuits that have been studied extensively in the past.  
  Nisan~\cite{Nis91} obtained an exact complexity characterization of ABPs in the non-commutative setting. 
  The problem of proving size  lower bounds for the general class of algebraic branching programs is widely open. When the ABP is restricted to be homogeneous, the best known lower bound is only quadratic in the number of variables~\cite{Kum17}. 
As ABPs are apparently more powerful than formulas but less powerful than circuits, proving lower bounds for syntactic multilinear ABPs (smABPs)
seems to be the next natural step towards the 
goal of proving super-polynomial lower bounds for syntactic multilinear circuits. Even in  the case of syntactic multilinear ABPs, no super-quadratic lower bound is known~\cite{Jan08}. 

\subsubsection*{ Models and results}
In this article, we are  interested in syntactic multilinear ABPs  and their  sub-classes where order of the appearance of  variables   along any path in the ABP is restricted.

To begin with, we give a decomposition theorem for smABPs. The decomposition obtains two disjoint sets $E_1$ and $E_2$ of edges in the branching program $P$ with source $s$ and sink $t$  such that the polynomial computed by it can be expressed as sum of $\sum_{(u,v)\in E_1}[s,u]\cdot \lab(u,v)\cdot [v,t]$ and $\sum_{(w,a)\in E_2}[s,w]\cdot \lab(w,a)\cdot [a,t]$ where $[p,q]$ is the polynomial computed by sub-program in $P$ with source $p$ and sink $q$. Also the sets $E_1$ and $E_2$ are chosen carefully such that the sub-programs obtained are more or less balanced in terms of the number of variables. More formally, we prove:


\begin{restatable}{theorem}{abpdecomp}
\label{lem:balance-interval}
Let $P$ be an smABP of size $S$ computing $f\in\mathbb{F}[x_1,\ldots,x_n]$. There exists edges $\{(u_1, v_1), \ldots, (u_m,v_m)\}$ and  $\{(w_1, a_1), \ldots, (w_r, a_r)\}$ in $P$ such that
\begin{enumerate}
\item[(1)] For $i\in [m]$, $ n/3 \le  |X_{s,u_i}|  \le 2n/3 $; and
\item[(2)] For $i\in [r]$, $|X_{s,w_i}| + |X_{a_i,t}| \le 2n/3$; and 
\item[(3)] $f= \sum_{i=1}^m [s,u_i]\cdot \lab(u_i,v_i)\cdot [v_i,t] +  \sum_{i=1}^r [s,w_i]\cdot \lab(w_i,a_i)\cdot[a_i,t]$.
\end{enumerate} 
\end{restatable}

 Let $\Sigma\Pi^{[\sqrt{n}]}(\Sigma\Pi)^{[\sqrt{n}]}$
denote the class of depth four  arithmetic circuits where the top layer of $\Pi$ gates  are products of at most $O(\sqrt{n})$  polynomials each  being a multilinear polynomial on $O(\sqrt{n})$ variables. As an immediate corollary of the above decomposition, we obtain the following low arity version of the depth reduction in~\cite{AV08,Tav13} for the case of  smABPs:

\begin{restatable}{coro}{depthredn}
\label{cor:depth4}
Let $P$ be a syntactic multilinear ABP of size $S$ computing a polynomial $f$ in $\mathbb{F}[x_1,\ldots,x_n]$. Then there exists a $\Sigma\Pi^{[\sqrt{n}]}(\Sigma\Pi)^{[\sqrt{n}]}$ syntactic multilinear formula  of size $2^{O(\sqrt{n}\log n \log S)}$  computing $f$.
\end{restatable}



Further, using the structural property of the parse trees of formulas obtained from smABPs, we prove exponential size lower bounds for two classes of smABPs with restrictions on the variable order.

{\em Strict circular-interval} ABPs are smABPs in which the index set of variables in every subprogram is contained in some circular interval in $\{1,\ldots,n\}$. (See Section~\ref{sec:circ-interval} for a formal definition). It may be noted that every multilinear polynomial can be computed by a strict circular-interval ABP  and hence it is a  universal model for computing multilinear polynomials.  We obtain an exponential lower bound on the size of any strict circular-interval ABP computing an explicit polynomial defined by Raz and Yehudayoff~\cite{RY08}. 

\begin{restatable}{theorem}{intervallb}
\label{thm:interval}
There exists an explicit multilinear polynomial $g$ in $\mathbb{F}[x_1,\ldots,x_n]$ such that any strict circular-interval ABP computing $g$ requires size $2^{\Omega(\sqrt{n}/\log n)}$. 
\end{restatable}





Yet another sub-class of smABPs that we study are the class of {\em bounded order}  smABPs.
Jansen \cite{Jan08} introduced  ordered algebraic branching programs. Ordered smABPs are smABPs with source $s$ and sink $t$ such that every path from $s$ to $t$ reads variables in a fixed order $\pi\in S_n$.  Jansen \cite{Jan08} translated the  exponential lower bound for the non-commutative model in~\cite{Nis91}  to ordered smABPs.  Ordered ABPs have also been studied in the context of the polynomial identity testing problem. Sub-exponential time algorithms were obtained for  identity testing of polynomials computed by ordered ABPs are known. (See \cite{JQS09,FSS14,GKS17} and the references therein.) Further, it is shown in~\cite{JQS09} that ordered smABPs are equivalent to read-once oblivious ABPs (ROABPs for short, see Section~\ref{sec:prelim} for a definition).



 A natural generalization for ordered smABPs is to allow multiple orders. A smABP is $\lo$-ordered  if variables can occur in one of the $\lo$ fixed orders along any source to sink path. In this article, we study $\lo$-ordered smABPs  and obtain structural results as well as an exponential lower bound for the model. 
We show the construction given in~\cite{JQS09} for the equivalence of ROABPs and 1-ordered ABPs can be generalized to $\lo$-ordered smABPs.  In particular, we prove that ${\lo}$-ordered ABP of size $S$ can be transformed into an equivalent ${\lo}$-pass smABP of size $\poly(S,\lo)$ (Theorem~\ref{thm:order-to-pass}). Though the overall idea is simple, the construction requires a lot of book-keeping of variable orders. Combining Theorem~\ref{thm:order-to-pass} with the lower bound   for sum of $k$-pass smABPs given in~\cite{CR18}, we get an exponential lower bound for sum of $\lo$-ordered  smABPs when $k=o(\log n)$.
By exploiting a simple structural property of $\lo$ ordered ABPs, we  prove an exponential lower bound for sum of $\lo$ ordered smABPs even when $\lo$ is sub-exponentially small:

\begin{restatable}{theorem}{orderedlb}
\label{thm:sum-ordered}
Suppose $\lo\leq 2^{n^{1/2-\epsilon}}$ for some $\epsilon>0$ and $g = f_1+f_2 +\dots + f_{m}$, where $f_i$ is computed by $\lo$-ordered ABP $P_i$ of size $S_i$. Then, either there is an $i\in [m]$ such that $S_i = 2^{\Omega(n^{1/40})}$ or $m = 2^{\Omega(n^{1/40})}$, where $g$ is the polynomial defined in~\cite{RY08}.
\end{restatable}


\paragraph*{Related works}
It may be noted that the depth reductions in~\cite{AV08,Tav13} also preserve syntactic multilinearity.  However, Corollary~\ref{cor:depth4} obtains a sum of products of low-arity polynomials, which is new. As far as we know,  none of the known depth reductions achieve this arity bound. 
Saptharishi~\cite[Chapter 18, Lemma~18.8]{Sap}, observes that proving  an exponential lower bound for low bottom fan-in $\Sigma\Pi\Sigma$ circuits is enough to separate ${\sf VP}$ from ${\sf VNP}$.  However, as far as we are aware,  the argument by Saptharishi~\cite{Sap} uses random restrictions to variables, and does not   lead directly to a depth reduction.

In~\cite{AR16} Arvind and Raja have considered  interval ABPs  where   for every node $v$  reachable from the source, the sub-program with $v$ as the sink node must have an interval as the variable set. They proved exponential size lower bound for interval ABPs assumming the sum of squares conjecture.  Our model though is more restrictive than the one in~\cite{AR16}, our lower bound argument is unconditional.  



\section{Preliminaries}
\label{sec:prelim}
In this section we include necessary definitions of  all models  and notations used.  Let $X= \{x_1,\ldots, x_n\}$ denote a finite set of variables.

An {\em arithmetic circuit}~$C$ over $\mathbb{F}$ is a directed acyclic graph with vertices of in-degree at most 2.  A vertex of out-degree 0 is called an output gate. 
 The vertices of in-degree 0 are called input gates and are labeled  by elements from $X \cup \mathbb{F}$. All internal gates are labeled either  by  $+$ or $\times$. Every gate in $\cal{C}$ naturally computes a polynomial and the polynomial computed at the output gate is the output of the circuit.  The {\em size} of an arithmetic circuit is the number of gates in $\mathcal{C}$ and {\em depth} of $\mathcal{C}$ is the length of the longest path from an input gate to the output gate in $\mathcal{C}$. An {\em arithmetic formula} is an arithmetic circuit where the underlying undirected graph is a tree. 

 A {\em parse tree } $T$ of an arithmetic formula $F$  is a sub-tree of $F$ containing the output gate of $F$ such that  for every $+$ gate $v$  of $F$ that is included in $T$, exactly one child of $v$ is in $T$ and for every $\times $ gate $u$ that is in $T$, both children of $u$ are in $T$.

  An {\em algebraic Branching Program} (ABP) $P$  is a directed acyclic graph with one vertex $s$ of in-degree 0 (source) and one vertex $t$ of out-degree 0 (sink). The vertices of the graph are partitioned into layers $L_0,L_1,\ldots,L_{\ell}$ where edges are from vertices in layer $L_i$ to those in $L_{i+1}$ for every $0\leq i\leq \ell-1$. The source node  $s$ is the only vertex in layer $L_0$ and the sink  $t$ is the only vertex in layer $L_{\ell}$.  Edges  in $P$  are labelled by  an element in $X \cup \mathbb{F}$ and let $\lab(e)$ denote the label of an edge $e$. The width of the ABP $P$ is $\max_{i}\{|L_i|\}$ and size of the ABP $P$ is the number of nodes in $P$.   Let weight of a path be the product of its edge labels.  The polynomial  computed by an ABP $P$  is the sum of weights of all $s$ to $t$ paths in $P$.  For  nodes $u$ and $v$ in $P$, let $[u,v]_P$ denote the polynomial computed by the sub-program of $P$ with $u$ as the source node and $v$ as the sink node. We drop the subscript from $[u,v]_P$ when $P$ is clear from the context. Let $X_{u,v}$ 
  denote the set of all variables that occur as labels in any path from $u$ to $v$ in $P$.  

An ABP $P$ is said to be {\em syntactic multilinear} (smABP)  if every  variable occurs at most once in every $s$ to $t$ path in  $P$.  
 $P$ is said to be an {\em Oblivious-ABP} if for every layer $L$ in $P$, there is a variable $x_{i_L}$ such that every edge from the layer $L$  is labeled from  $\{x_{i_L}\} \cup \mathbb{F}$.
An smABP $P$ is said to be {\em Read-Once Oblivious} (ROABP)  if  $P$ is oblivious  and every variable    appears as edge label in at most one layer.

Anderson et al.~\cite{AFSSV18} defined the class of $\lo$ pass smABPs.  An oblivious  smABP  $P$ is $\lo$ pass, if there are layers $i_1<i_2< \ldots< i_{\lo}$ such that for every $j$, between layers $i_j$ and $i_{j+1}$  the program $P$ is is an ROABP. 
%
Let $\pi$ be a permutation of $\{1,\ldots, n\}$ and $P$ be an smABP computing an $n$ variate multilinear polynomial. An $s$ to $t$ path $\rho$ in $P$ is said to be {\em consistent} with  $\pi$, if the variable labels in $\rho$ occur as per the order given by $\pi$, i.e, $x_i$ and $x_j$  occur as edge labels in $\rho$ in that order, then $\pi(i) < \pi(j)$. For a node  $v$ of $P$,  $v$ is said to be consistent with $\pi$, if every $s$ to $v$ path is consistent with $\pi$. 

An smABP  $P$ is said to be {\em $\lo$ ordered}, if there are $\lo$ permutations $\pi_1, \ldots, \pi_\lo$ such that for  every $s$ to $t$ path $\rho$ in $P$,  there is an $1\le i\le \lo$ such the $\rho$ is consistent with $\pi_i$.
 
We now review the partial derivative matrix of a polynomial introduced in~\cite{Raz09}. Let $\mathbb{F}$ be a field and $X=Y \cup Z$  be such that $Y\cap Z = \emptyset$ and $|Y| = |Z|$. It is convenient to represent the partition $X = Y\cup Z$ as an injective function $\varphi: X \to Y \cup Z$. For a polynomial $f$,  let $f^{\varphi}$ be the polynomial obtained by relabeling each  variable $x_i$ by $\varphi(x_i)$.

\begin{defn}{\em $($Partial Derivative Matrix.$)$}~\cite{Raz09}
Let $f\in\mathbb{F}[X]$ be a multilinear  polynomial. The {\em partial derivative matrix} of $f$ (denoted by $M_f$) with respect the partition $\varphi: X \to Y\cup Z$ is a $2^m\times 2^m$ matrix defined as follows. For  multilinear monomials $p$ and $q$ in variables $Y$ and $Z$ respectively, the entry $M_f [p,q]$ is the coefficient of the monomial $pq$ in $f^\varphi$.  
\end{defn}
 For a polynomial $f$ and a partition $\varphi$, let  $\rank(f)$ denote the rank of the matrix $M_f$ over the field $\mathbb{F}$. The following properties of $\rank(f)$ are useful: 

\begin{lemma}(\cite{Raz09}) Let $f,g \in \mathbb{F}[Y,Z]$.   We have:
\label{lem:rankub}
\label{lem:sub-aditivity}
\label{lem:sub-multiplicativity}
\begin{enumerate}
\item[(1)] $\rank({f+g}) \leq \rank(f)+\rank(g).$
\item[(2)] If $\var(f) \cap \var(g) = \emptyset$, then $\rank({fg}) = \rank(f) \cdot \rank(g)$.
 \item[(3)] If   $f\in\mathbb{F}[Y_1,Z_1]$ for $Y_1\subseteq Y, Z_1\subseteq Z$,  then $\rank(f)\leq 2^{\min\{|Y_1|,|Z_1|\}}$.
\end{enumerate}
\end{lemma}

\noindent Let ${\cal D}$ denote the uniform distribution on the set of all partitions $\varphi:  X \to Y \cup Z$, where $|Y| = |Z| = |X|/2$.
In~\cite{CR18}, it is shown that for any polynomial computed by an  ROABP, rank of the partial derivative matrix is small with high probability:
\begin{lemma}{\em (Corollary~1 in \cite{CR18})}
\label{lem:roabp}
Let $f$ be an $N$ variate multilinear polynomial computed by a syntactic mutlilinear ROABP of size $S$. Then 
$$ \Pr_{\varphi\sim {\cal D}}[ \rank(f) \le S^{\log N}2^{N/2 - N^{1/5}}] \ge 1-2^{-N^{1/5}}.$$
\end{lemma}

We need the  following  polynomial defined  in~\cite{RY08}:   

\begin{defn}[Full rank Polynomial]{\cite{RY08}}
\label{def:raz-poly}
Let $n\in\mathbb{N}$ be even and   $\mathcal{W} = \{w_{i,k,j} \}_{i,k,j\in[n]}$. For any two integers $i,j\in\mathbb{N}$, we define an interval $[i,j] = \{ k\in\mathbb{N}, i\leq k\leq j \}$. Let $|[i,j]| = j-i+1$,  $X_{i,j} = \{ x_p \mid p\in [i,j]\} $ and $W_{i,j}=\{ w_{i',k,j'}\mid i',k,j'\in[i,j] \}$. Let  $\mathbb{G}=\mathbb{F}(\mathcal{W})$, the rational function field. For every $[i,j]$ such that $|[i,j]|$ is even we define a polynomial $g_{i,j}\in\mathbb{G}[X]$  as 
 $g_{i,j}=1$ when   $|[i,j]|=0$  and 
 if $|[i,j]|>0$ then, {\small $g_{i,j }\triangleq (1+x_ix_j)g_{i+1,j-1} + \sum_{k}w_{i,k,j}g_{i,k}g_{k+1,j}.$}
where $x_k$, $w_{i,k,j}$ are distinct variables, $1\le k\le j$ and the summation is over  $k\in [i+1,j-2]$ such that  $|[i,k]|$ is  even.  Let $g\triangleq g_{1,n}$.
\end{defn}

It is known that   for  any partition $\varphi$, $\rank(g)$ is the maximum possible value :

\begin{lemma}{\em \cite[Lemma 4.3]{RY08}}
\label{lem:ry}
Let $n\in\mathbb{N}$ be even   and  $\mathbb{G}$ as above.  Let $g\in\mathbb{G}[X]$ be the polynomial in Definition \ref{def:raz-poly}. Then for any $\varphi\sim {\cal D}$, 
$\rank({g})= 2^{n/2}$.
\end{lemma}

\section{A variable-balanced decomposition for syntactic multilinear ABPs}
\label{sec:decomp}
In this section, we give a new decomposition theorem for smABPs.  The decomposition can be seen as a variable  balanced version of the well known decomposition of arithmetic circuits given by Valiant et al.~\cite{VSBR83} for the case of smABPs. In fact, we show that  a syntactically multilinear ABP can be divided  into sub-programs that are almost balanced in terms of the number of variables. 


\abpdecomp*
\begin{proof}
The proof is by a careful subdivision of the program $P$. 
We assume without loss of generality that $t$ is reachable from every node in $P$ and that that every node in $P$ has in-degree and out-degree  at most 2.   Consider the following  coloring procedure:
\begin{enumerate}
\item[(1)]  Initialize by coloring $t$ as blue. Repeat (2) until no new node is colored.
\item[(2)]  Consider  node  $u$ that is colored blue and  such that nodes $v$ and $w$ are uncolored, where $(v,u)$ and $(w,u)$ are the only edges
incoming to $u$.
For $a \in \{v,w\}$ do following : 
\begin{enumerate}
\item If $|X_{s,a}| > 2n/3$, then color $a$ as blue.
\item  If $n/3 \le |X_{s,a}| \le  2n/3$, then color $a$ as red.
\item If $|X_{s,a}| < n/3$, then color $a$ as green.
\end{enumerate}
\end{enumerate} 
At the end of the above coloring procedure we have the following:
\begin{enumerate}
\item For every node $u$ with incoming edges $(v,u)$ and $(w,u)$, if $u$ is colored blue then both $v$ and $w$ are colored. 
\item For every $s\rightsquigarrow t$ directed path $\rho$ in $P$,  exactly one of the following holds: 
\begin{enumerate}
\item $\rho$ has exactly one  edge $(v,w)$ such that $v$ is colored red and $w$ is colored blue. 
\item $\rho$ has exactly one edge $(v,w)$ such that $v$ is green colored and $w$ is colored blue.
\end{enumerate}
\item If a node $u$ is colored blue, then  every node $v$ reachable from $u$ must have color blue. 
\end{enumerate} 
Property 1 follows from the fact that a node $v$  is colored  if and only if there is an edge $(v,u)$ such that $u$ is colored blue.  For property 3,  clearly, a node $u$ is colored blue if and only of $|X_{s,u}|>2n/3$, thus every node reachable from a blue node is also colored blue. For property 2, let $\rho$ be a directed $s\rightsquigarrow t$ and $v$ be the first node along $\rho$ that is colored blue. Note  $|X_{s,s}|=0$, so $s$ cannot be colored blue. Clearly, every node that follows $v$ in $\rho$ is colored blue and $u \neq s$. Let $u$ be the node that immediately precedes $v$ in $\rho$, then clearly,  $u$ is either red or green. Uniqueness   follows from the fact that no node that precedes $u$ in $\rho$ is coloured blue and every node that succeeds $v$ in $\rho$ is colored blue, hence there cannot be another such edge.

Let $ E_{rb}=\{(u,v)\in P \mid \text{$u$ is colored red and $v$ is colored blue}\}$ and $ E_{gb}=\{(u,v)\in P \mid \text{$u$ is colored green and $v$ is colored blue}\}$.  Let $E_{rb} = \{(u_1, v_1), \ldots, (u_m,v_m)\}$ and $E_{gb} = \{(w_1, a_1), \ldots, (w_r, a_r)\}$ where $m ,r \le 2S$.
We now prove that sets $E_{rb}$ and $E_{gb}$ satisfy the required properties.\begin{itemize}
\item[(1)] For $i\in [m]$, since $(u_i,v_i)\in E_{rb}$, $u_i$ is colored red. By Step 2(b) of coloring procedure, $n/3\leq |X_{s,u_i}|\leq 2n/3$.
\item[(2)] For $i\in [r]$, since $(w_i,a_i)\in E_{gb}$, $w_i$ is colored green and $a_i$ is colored blue. By Step 2(c) of coloring procedure, $|X_{s,w_i}|<n/3$ and by Step 2(a), $|X_{s,a_i}|>2n/3$. Since $P$ is syntactic multilinear,  $|X_{s,a_i}|+|X_{a_i,t}|\leq n$ implying $|X_{a_i,t}|\leq n/3 $. Therefore, $|X_{s,w_i}|+ |X_{a_i,t}| \le 2n/3$.
\item[(3)] By Property 2, $s\rightsquigarrow t$ paths in $P$ are partitioned into paths that have exactly one edge in $E_{rb}$ and paths that have exactly one edge in $E_{gb}$. Therefore,
\begin{align*}
f &= \sum_{\rho:s \rightsquigarrow t } {\sf wt}(\rho) = \sum\limits_{\substack{\rho:s \rightsquigarrow t, ~\rho\cap E_{rb}\neq \emptyset}}{\sf wt}(\rho) + \sum\limits_{\substack{\rho:s \rightsquigarrow t ,~ \rho\cap E_{gb}\neq \emptyset  } } {\sf wt(\rho)} \\
&=  \sum_{i=1}^m [s,u_i]\cdot \lab(u_i,v_i)\cdot [v_i,t] +  \sum_{i=1}^r [s,w_i]\cdot \lab(w_i,a_i)\cdot[a_i,t].
\end{align*}
where ${\sf wt}(\rho)$ denotes the product of edge labels of path $\rho$. 
\end{itemize}
\end{proof}
The above decomposition allows us to obtain small depth formulas for syntactic multilinear ABPs with quasi-polynomial blow up in size.  In the following, we show that a syntactic multilinear ABP can be computed by a $\log$-depth syntactic multilinear formula were each leaf represents a multilinear polynomial on $O(\sqrt{n})$ variables. 
\begin{lemma}
\label{lem:formula-interval}
Let $P$ be a syntactic multinear  ABP of size $S$ computing a multilinear polynomial $f$ in $\mathbb{F}[x_1,\ldots,x_n]$. Then, there is a syntactic multilinear formula $\Phi$ computing  $f$ of size $S^{O(\log n)}$ and depth $O(\log n)$ such that every leaf $w$ in $\Phi$ represents a multilinear polynomial $[u,v]_{P_w}$ for some nodes $u,v$ in $P_w$ with $ |X_{u,v}| \le \sqrt{n}$, where $P_w$ is a subprogram  of $P$. Further, any parse tree of $\Phi$ has at most $3\sqrt{n}$ leaves.  
\end{lemma}

\begin{proof}
Let nodes $s$ and $t$ be source and sink of $P$. The proof constructs a formula $\Phi\def \Phi_{s,t}$ by induction on the number of variables $|X_{s,t}|$  in the program $P$. 

\noindent{\bf Base Case :} If $|X_{s,t}| \le \sqrt{n}$, then   $\phi_{s,t}$ is a leaf gate with label $[s,t]$.

\noindent{\bf Induction Step :} For induction step, suppose $|X_{s,t}|>\sqrt{n}$. By Theorem~\ref{lem:balance-interval}, we have
\begin{align}
f= \sum_{i=1}^m [s,u_i]\cdot \lab(u_i,v_i)\cdot [v_i,t] +  \sum_{i=1}^r [s,w_i]\cdot \lab(w_i,a_i)\cdot[a_i,t].
\label{eq:divide-interval} 
\end{align} 
where $u_i, v_i, w_i$ and $a_i$ are nodes in $P$, with $|X_{s,t}|/3 \le |X_{s,u_i}|  \le 2|X_{s,t}|/3$ and $|X_{s,w_i}|+|X_{a_i,t}| \le 2|X_{s,t}|/3$. Further,  $[s,u_i]\cdot \lab(u_i,v_i)$ (resp. $[s,w_i]\cdot\lab{(w_i,a_i)}$) is an smABP  with at most $2|X_{s,t}|/3+1$ (resp. $|X_{s,t}|/3$) variables.  Let 
\begin{equation}
\label{eqn:decom}
f= \sum_{i=1}^m  g_i h_i + \sum_{i=1}^r g_i'h
_i'.
\end{equation}
 where $g_i = [s,u_i]\cdot \lab(u_i,v_i) $, $h_i = [v_i,t]$, $g_i' = [s,w_i]\cdot \lab(w_i,a_i)$ and $h_i' = [a_i,t]$. 
For any $i$,  if $|X_{s,w_i}| + |X_{a_i,t}| < \sqrt{n}$, then we set $g_i' = [s,w_i]\cdot \lab(w_i,a_i) \cdot [a_i,t]$  and $h_i' =1$. By induction, suppose $\phi_i$ (resp. $\phi_i'$) be the multilinear formula that computes $g_i$ (resp. $g_i'$) and $\psi_i$ (resp. $\psi_i'$)  be that for $h_i$ (resp. $h_i'$). Set $\Phi = \sum_{i=1}^m (\phi_i\times \psi_i) + \sum_{i=1}^r (\phi_i'\times \psi_i') $. Let $T(n)$ denote the size of the resulting formula on $n$ variables. Then, $T(n)
\le 2\cdot S\cdot 2\cdot T(2n/3) = S^{O(\log n)}$. Thus, $\Phi$ is a syntactic mutlilinear formula of size $S^{O(\log n)}$ and depth $O(\log S)$ computing $f$ and by construction  every leaf represents a multilinear polynomial $[u,v]_P$ for some nodes $u,v$ in $P$ with $ |X_{u,v}| \le \sqrt{n}$.

It remains to prove that  any parse tree of $\Phi$ has at most $3\sqrt{n}$ leaves. 
 We begin with a description of the process for constructing parse sub-trees of $\Phi$. By Equation (\ref{eqn:decom}), constructing a parse tree of $\Phi$ is equivalent to the  process:
\begin{itemize}
\item[1.] Choose $b\in\{0,1\}$ (corresponds to choosing one of the  summations in Equation (\ref{eqn:decom})).
\item[2.] If $b=0$ choose $i\in \{1,\ldots, m\}$, else if $b=1$ choose $j\in \{1,\ldots, r\}$.
\item[3.] Repeat steps 1 and 2 for sub-formulas $\phi_i,\phi_i',\psi_i$ and $\psi_i'$.
\end{itemize}

Consider any parse tree $T$ of $\Phi$. 
It is enough to prove that every leaf in $T$ that is not labeled by $1$ is a polynomial in at least $\sqrt{n}/3$ variables.  Since $\Phi$ is syntactic multilinear any parse tree of $\Phi$ has at most $3\sqrt{n}$ leaves as required. However, it  can be noted that this may not be true always. Instead, we argue that every leaf in $T$ can be associated with a set of at least $\sqrt{n}/3$ variables such that no other leaf in $T$ can be associated with these variables, hence implying that the number of leaves in any parse $T$ of $\Phi$ is at most $3\sqrt{n}$.

Consider a leaf $v$ in $T$ having less than  $\sqrt{n}/3$ variables. Let $u$ be the first sum gate on the path from $v$ to root with $|X_u|>\sqrt{n}$. Rest of the argument is split based on  whether $b=0$ or $b=1$ at the step for choosing $v$ in the construction of parse tree $T$. Throughout this proof for any gate $u$ $X_u$  denotes the set of variables in the sub-circuit rooted at $u$ in $\Phi$.


First suppose that  in the construction of $T$,  $b=0$ at the step for choosing $v$. Then,  either $v=[p,u_i]\cdot \lab(u_i,v_i)$ or $v=[v_i,q]$ for some nodes $p, q, u_i,v_i$ in  $P$, where $u_i$(respectively $v_i$) is colored red(respectively blue) when the coloring procedure is performed on the sub-program with source $p$ and sink $q$. 
 If $v=[p,u_i]\cdot \lab(u_i,v_i)$, $|X_v|\geq |X_{p,u_i}|\geq |X_u|/3 \geq \sqrt{n}/3$, a contradiction to  fact that $v$ is a leaf in $T$ with fewer than $\sqrt{n}/3$ variables. Hence, $v=[v_i,q]$. Set $A(v) = X_u \setminus (X_{p,u_i} \cup \{\lab(u_i,v_i)\})$, clearly $|A(v)| \ge \sqrt{n} /3$, as $|X_u| \ge \sqrt{n}$ and  $|X_{p,u_i}| \le 2|X_u|/3$.

When the $b=1$, we have the following possibilities:
\begin{description}
\item[Case 1]  $v = [p, w_i]\cdot \lab(w_i, a_i)\cdot [a_i,q]$. In this case, set $A(v) = X_u$. Then $|A(v)|\geq \sqrt{n}/3$.  
\item[Case 2]  $v = [p, w_i]\cdot \lab(w_i, a_i)$. In this case, set $A(v) = X_u \setminus \var([a_i,q])$. 

Then $|A(v)|=|X_u|-|\var([a_i,q])| \geq \sqrt{n}/3$ as $|X_{a_i,q}|\leq 2|X_u|/3$ and $|X_u|>\sqrt{n}$. 
\item[Case 3] $v = [a_i,q]$. Set $A(v) = X_u \setminus (\var([p,w_i]) \cup \{\lab(w_i,a_i)\})$. Then $|A(v)|=|X_u|-|\var([p,w_i])| \geq \sqrt{n}/3$ as $|X_{p,w_i}|\leq 2|X_u|/3$ and $|X_u|>\sqrt{n}$.
\end{description}


It remains to prove that, for any two distinct leaves $v$ and $v'$ in $T$ such that $A(v)$ and $A(v')$ are defined, $A(v) \cap A(v') = \emptyset$. Let $u$ and $u'$ respectively be parents of $v$ and $v'$ in $T$.

When $u=u'$, either $v=[p, w_i]\cdot \lab(w_i, a_i)\cdot [a_i,q],v'=1$ or vice-versa or $v=[p, w_i]\cdot \lab(w_i, a_i) ,v'=[a_i,q]$ or vice-versa. As $A(v)$ is defined only for non-constant leaves, the only case is when $v=[p, w_i]\cdot \lab(w_i, a_i) ,v'=[a_i,q]$ or vice-versa. In any case, we have $A(v)\cap A(v')=\emptyset$.  Now suppose, $u\neq u'$ and   $A(v) \cap A(v') \neq \emptyset$.  Then, we have we have $X_u \cap X_{u'} \neq \emptyset$ as $A(v)\subseteq X_u$ and $A(v')\subseteq X_{u'}$.  From the  fact that $u,u'$ appear in the same parse tree we can conclude the least common ancestor of $u$ and $u'$ in $\Phi$ must be a $\times$ gate.  Let $[p,q]$ and $[p'q']$ be the sub-programs of $P$ that correspond to $u$ and $u'$ respectively. By the construction of $\Phi$, we can conclude that either there is a path from $q$ to $p'$ or there is a path from $q'$ to $p$ in $P$. Either of the cases is a contradiction to the fact that $P$ is syntactic multilinear. 
\end{proof}

Now, we obtain a reduction to depth-4 formulas for syntactic multilinear ABPs. Denote by $\Sigma^{[T]}\Pi^{[d]}(\Sigma\Pi)^{[r]}$ the class $\Sigma_{i=1}^{T}\Pi_{j=1}^{d}Q_{ij}$ where $Q_{ij}$'s are mulitlinear polynomials in  $O(r)$ variables. As a corollary to Lemma \ref{lem:formula-interval} we have the following  reduction to syntactic multilinear  $\Sigma\Pi^{[\sqrt{n}]}(\Sigma\Pi)^{[\sqrt{n}]}$ formulas for smABPs.


\depthredn*

\section{Strict Circular-Interval ABPs}
\label{sec:circ-interval}
In this section we prove an exponential size lower bound against a special class of smABPs that we call as $\cintabp$s. 


An interval $I=[i,j]$ in $\{1,\ldots,n\}$ is a circular $\pi$-interval if $I=\{\pi(i),\pi(i+1),\ldots,\pi(j)\}$ for some $i,j\in [n],i<j$ or $I=\{\pi(i),\pi(i+1),\ldots,\pi(n),\pi(1),\ldots,\pi(j)\}$ for some $i,j\in [n],i>j$. These intervals are called {\em circular intervals} as every such interval $[i,j]$ in $\{1,\ldots,n\}$ can be viewed as a chord on the circle containing $n$ points. Two circular intervals $I$ and $J$ are said to be {\em overlapping} if the corresponding chords in the circle intersect and {\em non-overlapping} otherwise.

We define a special class of syntactic multilinear ABPs where every  the set of variables involved in every subprogram is in some  $\pi$ circular interval. 

\begin{defn}[Strict Circular-Interval ABP] Let $\pi\in S_n$ be a permutation. 
A syntactic mulitlinear ABP $P$ is said to a {\em strict $\pi$-circular-interval ABP} if
\begin{enumerate}
\item For any pair of  nodes $u,v$ in $P$, the index set of $X_{u,v}$ is contained in some circular $\pi$-interval $I_{uv}$ in $[1,n]$; and
\item  For any $u,a,v$ in $P$, the circular $\pi$-intervals $I_{ua}$ and $I_{av}$ are non-overlapping.
\end{enumerate}
$P$ is said to be {\em strict circular-interval ABP} if it is   strict $\pi$-circular-interval ABP for some permutation $\pi$.
\end{defn} 

\subsection*{Lower bound for strict circular-interval ABPs}
In this section, we obtain an exponential lower bound against strict circular-interval ABPs.  We require a few preliminaries:

\begin{itemize}
\item[1.] For every permutation $\pi$ in $S_n$, define the partition function $\varphi_\pi: X\rightarrow Y \cup Z$ such that
for all $1\leq i\leq n/2$, $\varphi(x_{\pi(i)})=y_i$ and $\varphi(x_{\pi(n/2+i)})=z_i$. 
\item[2.]For any $\pi$ in $S_n$, $|\varphi_\pi(X)\cap Y| = |\varphi_\pi(X)\cap Z|=|X|/2$. For  polynomial $g$ in Definition \ref{def:raz-poly}, ${\sf rank}_{\varphi_\pi}(g)= 2^{n/2}$ by Lemma \ref{lem:ry}.
\item[3.]For any set $X_i\subseteq X$, let $\varphi_\pi(X_i) = \{\varphi_\pi(x) \mid x\in X_i \}$. We say $X_i$ is {\em monochromatic} if either $\varphi_\pi(X_i)\cap Y = \emptyset$ or  $\varphi_\pi(X_i)\cap Z = \emptyset$. Observe that if $X_i$ is monochromatic then for any polynomial  $p_i\in \mathbb{F}[X_i]$, we have ${\sf rank}_{\varphi_\pi}(p_i)\leq 1$. Further, we say set $X_i\subseteq X$ is {\em bi-chromatic} if $\varphi_\pi(X_i)\cap Y \neq \emptyset$ and  $\varphi_\pi(X_i)\cap Z \neq \emptyset$.
\end{itemize}


%
%
%

%
In the following, we show that for any strict circular-interval ABP $P$ computing a polynomial 
$f$, there is a partition $\varphi$ such that  $\rank(f) $ is small:
\begin{theorem}
\label{thm:rank-ub-interval}
Let $P$ be a strict circular-interval ABP of size $S$ computing $f$ in $\mathbb{F}[x_1,\ldots,x_n]$. There exists a $\varphi: X\rightarrow Y \cup Z$ with $|\varphi(X) \cap Y|=|\varphi(X)\cap Z|=|X|/2$ such that $\rank(f) \le 2^{\sqrt{n} \log n \log S } 2^{\sqrt{n}}$. 
\end{theorem}

\begin{proof}
Let $\Phi$ be the syntactic multilinear formula  constructed from $P$ as given by Lemma~\ref{lem:formula-interval}. Note that any parse tree of $\Phi$ has at most $3\sqrt{n}$ leaves. The number of  parse trees of $\Phi$ is at most $\binom{2^{O({\log n\log S})}}{3\sqrt{n}} \leq 2^{\epsilon\log n \log S \sqrt{n}}$. Let $T$ be any parse tree of $\Phi$ with leaves $w_1,\ldots,w_\ell$ computing polynomials $p_1,\ldots,p_{\ell}$. We have $f=\sum_{T:\text{parse tree of $\Phi$}} m(T)$ where $m(T)$ be the product of multilinear polynomials corresponding to the leaves of $\Phi$ in $T$. Let $X_1,\ldots,X_\ell \subseteq X$ be such that $p_i$ is a polynomial in $\mathbb{F}[X_i]$. For every $i\in [\ell]$, let $M_i=\{ j \mid x_j\in X_i\}$ be the index set of $X_i$.  As $P$ is a strict circular-interval ABP, we have that sets $M_1,\ldots,M_\ell$ are   circular $\pi$-intervals in $\{1,\ldots,n\}$ for some $\pi\in S_n$. Let $\varphi_{\pi}: X \rightarrow Y \cup Z$ be the partition function  described above. 
If $X_i$ is bi-chromatic then ${\sf rank}_{\varphi_\pi}(p_i)\leq 2^{\sqrt{n}/2}$ as $|X_i|\leq \sqrt{n}$ by construction of formula $\Phi$ when $w_i$ is a leaf in $\Phi$. 

A crucial observation is that for any parse tree $T$ of $\Phi$, at most two of $\varphi_\pi(X_1),\ldots,\varphi_\pi(X_\ell)$ are bi-chromatic.  This is because the existence of bi-chromatic sets $\varphi_\pi(X_i),\varphi_\pi(X_j),\varphi_\pi(X_k)$ for some $i,j,k\in [\ell]$ implies that the circular $\pi$-intervals $M_i,M_j,M_k$ are overlapping from the way partition $\varphi_\pi$ is defined. As $X_i,X_j,X_k$ are variable sets associated with leaves of the same parse tree $T$, we can conclude that when 
$\varphi_\pi(X_i),\varphi_\pi(X_j),\varphi_\pi(X_k)$ are bi-chromatic there exists nodes $u,a,v$ in $P$ such that circular $\pi$-intervals $I_{ua}$ and $I_{av}$ are overlapping, a contradiction to the fact that $P$ is a strict circular-interval ABP.

Therefore, in any parse tree $T$ of $\Phi$, at most two of $\varphi_\pi(X_1),\ldots,\varphi_\pi(X_\ell)$ are bi-chromatic say $\varphi_\pi(X_i)$ and $\varphi_\pi(X_j)$. Hence ${\sf rank}_{\varphi_\pi}(p_i)\leq 2^{\sqrt{n}/2}$ and ${\sf rank}_{\varphi_\pi}(p_j)\leq 2^{\sqrt{n}/2}$. Also, ${\sf rank}_{\varphi_\pi}(p_k)\leq 1$ for all $k\neq i,j$. Thus, ${\sf  rank}_{\varphi_\pi}(f) \le 2^{\epsilon\log n \log S \sqrt{n}} 2^{\sqrt{n}}$.
\end{proof}

With the above, we can prove Theorem~\ref{thm:interval}:

\intervallb*

\begin{proof}
Let $P$ be a strict circular-interval ABP of size $S=2^{o(\sqrt{n} /\log n)}$ computing $g$ and $\Phi$ be the syntactic multilinear formula obtained from $P$ using Lemma~\ref{lem:formula-interval}.  By Theorem \ref{thm:rank-ub-interval}, there exists a partition $\varphi:X \rightarrow Y\cup Z$ with $|\varphi(X)\cap Y|=|\varphi(X)\cap Z|=|\varphi(X)|/2$ such that $\rank(g) \leq 2^{\sqrt{n} + \epsilon\log n \log S \sqrt{n} } < 2^{n/2}$. However, by Lemma \ref{lem:ry}, $\rank(g)=2^{n/2}$, a contradiction. Hence $s=2^{\Omega(\sqrt{n} /\log n)}$.
 \end{proof}
Before concluding the section, we   observe that the arguments above imply a separation between  models that are closely related to strict circular interval ABPs.



\section{$\lo$-ordered ABPs}
\label{sec:ordered}

In this section we consider the case of $\lo$-ordered syntactic multilinear ABPs 
Ordered syntactic multilinear ABPs are well studied in the literature~\cite{Jan08,JQS09}. It is known that $1$-ordered syntatic multilinear ABPs are equivalent to syntactic multilinear ROABPs~\cite{JQS09}. In  Section~\ref{sec:pass}, we show that this result can be generalized to $\lo$-ordered ABPs where the resulting ABP makes at most $\lo$-passes on the variables, although in different orders. 
Further, in Section~\ref{sec:order-lb} we obtain an exponential lower bound for the sum of  $\lo$-ordered syntactic multilinear ABPs, where $\lo = 2^{n^{\frac{1}{2}-\epsilon}}$ for a small   constant $\epsilon>0$. 
\subsection*{${\cal L}$-ordered to ${\cal L}$-pass}
\label{sec:pass}
In this section, we show that $\lo$-ordered ABPs can be transformed into ABPs that make at most $\lo$-passes on the input, although in different orders.

\begin{theorem}
\label{thm:order-to-pass}
Let $P$ be an $\lo$-ordered ABP of size $S$ computing a polynomial $f\in\mathbb{F}[x_1,\ldots,x_n]$. Then there is an $\lo$-pass ABP $Q$ of size $\poly(\lo,S)$ computing $f$.
\end{theorem}
\begin{proof}
Let $P$ be an $\lo$-ordered ABP of size $S$ computing a polynomial $f$. Let $L_0,L_1,\ldots,L_{\ell}$ be the layers of $P$ where source $s$ and sink $t$ are the only nodes in layers $L_0$ and $L_\ell$ respectively. Let  $u_{i1},\ldots,u_{iw}$ be nodes in $L_i$, where $w\le S$ is the width of $P$.  Without loss of generality assume every node in $P$ has in-degree and out-degree at most two, and every layer except $L_0$ and $L_{\ell}$ has exactly $w$ nodes. Also, every $s$ to $t$ path in $P$ respects one of the permutations $\pi_1,\pi_2,\ldots,\pi_{\cal L}$. We now construct an ${\cal L}$-pass ABP ${\cal Q}$ that reads variables in the order $(x_{\pi_1(1)},x_{\pi_1(2)},\ldots,x_{\pi_1(n)}),\ldots,(x_{\pi_\lo(1)},x_{\pi_\lo(2)},\ldots,x_{\pi_\lo(n)}).$
The source and sink of ABP ${\cal Q}$ are denoted by $s'$ and $t'$ respectively.  The number of layers in ${\cal Q}$ will be bounded  by $\lo (\ell+1)$  and are labeled as  $L_{ir}, i\in [\lo],r\in \{0,\ldots,\ell\}$.  Intuitively,   for a  node $u_{rj}$ in layer $L_r$ in $P$, we have  $\lo$ copies, $u_{1rj},u_{2rj}\ldots, u_{\lo rj}$ in ${\cal Q}$, where $u_{irj}$ is a vertex in layer $L_{ir}$. Intuitively, $ u_{irj}$ would have all paths from $s$ to $u_{rj}$ that respect the permutation $\pi_i$, but none of the permutations $\pi_p$ for $p<i$.    To ensure that the resulting ABP is $\lo$-pass, we place the layers as follows : $L_{11}, \dots, L_{1\ell}, L_{21},\dots, L_{2\ell}, \dots, L_{\lo 1},\dots, L_{\lo \ell}.$

  We construct $Q$ inductively as follows : \begin{enumerate}
\item[(1)] \textbf{Base Case :} In ABP $P$, for every edge $e$  from source $s$ in layer $L_0$ to node $u_{1j}, j\leq w$ in layer $L_1$ labeled by $\lab(e)\in X\cup \mathbb{F}$, if $\lab(e)=x_k$, then add  the edge   $(s',u_{m1j})$ with label $x_k$   where $m$ is the smallest value such that $x_k$ is consistent with $\pi_m$,  if $\lab(e)=\alpha\in \mathbb{F}$, then add  the edge $(s', u_{m1j})$ with label $\alpha$.
\item[(2)] \textbf{Induction Step :} Consider layer $L_r,r\in\{1,\ldots,\ell\}$:
\begin{enumerate}
\item For every  node $u_{rj}$ in layer $L_r$ of $P$, with $1\le j \le w$ and  every  edge  $e$ of the form   $e=(u_{rj}, u_{r+1, j'})$ do the following:
\begin{description}
\item[Case~1:]  $\lab(e) =x_k \in X$. For every $1\le i\le \lo$, let $m$ be the smallest index such that every path from $s'$ to $u_{irj}$ concatenated with the edge $e$ is consistent with $\pi_{m}$. Note that, by the construction, $m\ge i$. Add  the  edge  $(u_{irj}, u_{mr+1j'})$ in ${\cal Q}$ for every $i$ with label $x_k$. For every $1\leq i \leq {\cal L}$, note that the choice of $m$ is unique.
\item[Case~2:] $\lab(e) =\alpha \in \mathbb{F}$. For every $1\le i \le \lo$, add edge $(u_{irj}, u_{ir+1j})$ with label $\alpha$.
\end{description}

\item[(b)] Create the node $t'$ in ${\cal Q}$, and add edges $(u_{i\ell 1}, t')$ with label $1$ for every $1\le i\le \lo$. 
\end{enumerate}
\end{enumerate}

\noindent Note that in the above construction, the resulting branching program will not be layered. It can be made layered by adding suitable new vertices and edges labeled by $1 \in \mathbb{F}$. 
\begin{claim}
\label{claim:order-to-pass}
We now prove the following:
\begin{enumerate}
\item[(1)] $Q$ is an  $\lo$-pass syntactic multilinear ABP and has size $\poly(\lo,S)$.
\item[(2)]  For $1\le r \le \ell$ and node $u_{rj}$ in layer $L_r$ in $P$, $1\le j\le w$, $[s,u_{rj}]_P = \sum_{i=1}^{\lo} [s',u_{irj}]_{\cal Q}.$
\end{enumerate}

\end{claim}
{\em Proof of Claim \ref{claim:order-to-pass}} : Clearly, for each $1\le i\le \lo$, the layers $L_{i1}, \ldots , L_{i\ell}$ in ${\cal Q}$ do a single pass on the variable in the order $\pi_i$.  Therefore,  ${\cal Q}$ is an  $\lo$-pass multilinear ABP reading variables in the order
$$(x_{\pi_1(1)},x_{\pi_1(2)},\ldots,x_{\pi_1(n)}),(x_{\pi_2(1)},x_{\pi_2(2)},\ldots,x_{\pi_2(n)}),\ldots,(x_{\pi_\lo(1)},x_{\pi_\lo(2)},\ldots,x_{\pi_\lo(n)}).$$

We prove (2) using induction on $r$.  For $r=0$, the statement follows immediately from the construction. Suppose, $r>0$.  
We have 
\begin{align*}
 [s, u_{rj}]_P &= \sum_{e = (u_{r-1j'}, u_{rj})} \lab(e) [s,u_{r-1j'}]_P =  \sum_{e = (u_{r-1j'}, u_{rj})} \lab(e)\cdot  \sum_{i=1}^\lo [s', u_{ir-1j'}]_{\cal Q} \\
            &=\sum_{e = (u_{r-1j'}, u_{rj})}\sum_{i=1}^\lo \lab(u_{ir-1j'}, u_{mrj})[s',u_{ir-1j'}]_{\cal Q} 
            = \sum_{i=1}^\lo [s',u_{irj}]_{\cal Q}.
          \end{align*}
          In the above, $m$ is the index as defined in the construction of ${\cal Q}$. \qedhere

\end{proof}

\subsection*{Lower bound for sum of $\lo$-ordered ABPs}
\label{sec:order-lb} 
It may be noted that the transformation of $\lo$-ordered ABPs to $\lo$-pass ABPs combined with the lower bounds for $\lo$-pass ABPs given in~\cite{CR18}, we have exponential lower bounds for sum of $\lo$-ordered ABPs when $\lo$ is bounded by $o(\log n)$.  In this section, we show that by observing a simple property of the  ABP to formula conversion given in Lemma~\ref{lem:formula-interval},  we can obtain lower bounds for $\lo$-ordered ABPs for larger sub-exponential values $\lo$. 
In the following, we observe that in the formula obtained using Lemma~\ref{lem:formula-interval} obtained from an $\lo$-ordered ABP, a lot of the leaves in any parse tree are 
in fact $1$-ordered ABPs:
\begin{lemma}
\label{lem:order-formula}
Let $P$ be an $\lo$-ordered ABP and $F$ be the syntactic multilinear formula obtained from $P$ using Lemma~\ref{lem:formula-interval}. 
Then, for any  parse tree $T$ of $F$, all but at most $O(\log \lo)$ many leaves of $T$ are $1$-ordered ABPs (ROABPs).
\end{lemma}

%

\begin{proof}
Let $T$ be any parse tree of $F$ with leaves $w_1,\ldots,w_{\ell}$ and and let $p_1, \ldots, p_{\ell}$ be the polynomials labeling $w_1,\ldots,w_{\ell}$.  From the construction given in the proof of Lemma~\ref{lem:formula-interval}, corresponding to each leaf $w_i$ there are nodes $u_i, v_i$ in $P$ such that polynomial $p_{i} = [u_i,v_i]\cdot \lab{(v_i, u_{i+1})}$.  Consider the syntactic multilinear ABP $P'$  obtained by placing programs $$[u_1,v_1]\cdot \lab{(v_1, u_{2})}, [u_2,v_2]\cdot \lab{(v_2, u_{3})}, \dots, [u_i,v_i]\cdot \lab{(v_i, u_{i+1})}, \dots , [u_\ell,v_\ell]
$$ in the above order {and identifying nodes appropriately}. From the construction above,  $P'$ is a  sub program of $P$ and hence the number of  variable orders in $P'$ is a lower bound on the  number of  variable orders in $P$. If $r_i$ is the number of variable  orders in the sub program $[u_i, v_i]$, the total number of variable orders in the sub program $P'$ (and hence $P$) is at least $r_1\cdot r_2 \cdots r_{\ell}$. Since the number of distinct orders is at most $\lo$, we  conclude that $|\{i~|~ r_i \ge 2\}| \le \log \lo$, as required.
\end{proof}
For the remainder of the section, let ${\cal D}$ denote the uniform distribution on the set of all partitions $\varphi:X \to Y\cup Z$ with $|Y| = |Z|.$
\begin{lemma}
\label{lem:rank-ordered} Let $P$ be an $\lo$-ordered ABP of size $S$ computing a polynomial $f$. Then for $k = n^{1/20}$,
$\Pr_{\varphi\sim {\cal D}}[\rank(f)> 2^{\log n \log S \sqrt{n}}\cdot 2^{n/2 - k\sqrt{n}}] \le 2^{-O(n^{1/20})}.$
\end{lemma}

Now, we are ready to prove Theorem~\ref{thm:sum-ordered}:

\orderedlb*

\begin{proof}
Set $k = n^{1/32}$. Suppose, for every $i$, $f_i$ is computed by $\lo$-ordered ABP of size $2^{n^{1/40}}$. Then $\rank(f_i) > 2^{(\log n\log (2^{n^{1/40}} )\sqrt{n}} 2^{n/2 - k\sqrt{n}}$  with probability at most 
 $2^{2n^{1/40}}2^{-n^{1/20}}$  when $\varphi\sim {\cal D}$.  Therefore, probability that there is a $i$ such that  $\rank(f_i) > 2^{(\log n\log S )\sqrt{n}} 2^{n/2 - k\sqrt{n}}$ is at most $ m 2^{2n^{1/40}}2^{-n^{1/20}}< 1$ for $m<2^{n^{1/40}}$.  By union bound, there is a $\varphi\sim {\cal D}$ such that for every $i$,  $\rank(f_i) < 2^{(\log n\log (2^{n^{1/40}} ))\sqrt{n}} 2^{n/2 - k\sqrt{n}} < 2^{n/2}$.  But $\rank(g)= 2^{n/2}$ for every partition $\varphi$, which is a contradiction.  
\end{proof}



\bibliographystyle{plain}
\bibliography{refbib}

\end{document}